  \documentclass[multi]{cambridge6Atight}

  \usepackage{natbib}
  \usepackage{rotating}
  \usepackage{floatpag}
  \rotfloatpagestyle{empty}

  \usepackage{amsmath}
  \usepackage{amsthm,amssymb}
  \usepackage{graphicx}

\usepackage[utf8]{inputenc}
\usepackage{bbm}

\usepackage{graphicx}
\usepackage{pgfplots}
\usepackage{tikz}
\usetikzlibrary{shapes, patterns, decorations, fit, intersections}

\newcommand{\idf}[2]{\textit{#1}\index{#2}}

\renewcommand{\index}[1]{}

\newcommand*{\one}{\mathbbm{1}}

\newcommand*{\abs}[1]{\lvert#1\rvert}

\newtheorem{theorem}{Theorem}
\newtheorem{lemma}[theorem]{Lemma}
\newtheorem{corollary}[theorem]{Corollary}

\newtheorem{claim}[theorem]{Claim}

\newtheorem{fact}[theorem]{Fact}

\theoremstyle{definition}
\newtheorem{definition}[theorem]{Definition}
\newtheorem{remark}[theorem]{Remark}

\usepackage{slivkins-setup,ch-macro,nicefrac}

\usepackage{url,hyperref}

\usepackage{color-edits}
\addauthor{as}{red}      

\usepackage[linesnumbered,lined,boxed,algochapter]{algorithm2e}
\SetKwInOut{Parameter}{parameter}

\begin{document}




\alphafootnotes
\author[Aleksandrs Slivkins]{Aleksandrs Slivkins}
\chapter{Exploration and Persuasion}



\copyrightline{
Available from the author's website since 2021.
\newline
Authors' address and email: Microsoft Research, New York, NY; slivkins@microsoft.com.
\vspace{2mm}\newline
\textbf{Acknowledgements.} The author is grateful to Ian Ball, Brendan Lucier and David Parkes for careful reading of the manuscript and valuable suggestions. Sections 6-8 are based on \citep[Ch. 11]{slivkins-MABbook}.
\vspace{2mm}\newline
[Exploration and Persuasion \copyright Aleksandrs Slivkins 2021] has been published in 2023 a chapter in \emph{Online and Matching-Based Markets} \copyright Cambridge University Press. Not for distribution.}

\arabicfootnotes

\vspace{-20mm}

\begin{abstract}
\textsl{How to incentivize self-interested agents to explore when they prefer to exploit?}

Consider a population of self-interested agents that make decisions under uncertainty. They \emph{explore} to acquire new information and \emph{exploit} this information to make good decisions. Collectively they need to balance these two objectives, but their incentives are skewed toward exploitation. This is because exploration is costly, but its benefits are spread over many agents in the future.

\emph{Incentivized Exploration} addresses this issue via strategic communication. Consider a benign ``principal" which can communicate with the agents and make recommendations, but cannot force the agents to comply. Moreover, suppose the principal can observe the agents' decisions and the outcomes of these decisions. The goal is to design a communication and recommendation policy which (i) achieves a desirable balance between exploration and exploitation, and (ii) incentivizes the agents to follow recommendations. What makes it feasible is \emph{information asymmetry}: the principal knows more than any one agent, as it collects information from many. It is essential that the principal does not fully reveal all its knowledge to the agents.

Incentivized exploration combines two important problems in, resp., machine learning and theoretical economics. First, if agents always follow recommendations, the principal faces a \emph{multi-armed bandit} problem: essentially, design an algorithm that balances exploration and exploitation. Second, interaction with a single agent corresponds to \emph{Bayesian persuasion}, where a principal leverages information asymmetry to convince an agent to take a particular action. We provide a brief but self-contained introduction to each problem through the lens of incentivized exploration, solving a key special case of the former as a sub-problem of the latter.


The chapter is structured as follows:
\begin{OneLiners}
\item[1.] Motivation and problem formulation;
\item[2.] Connection to multi-armed bandits and optimal exploration for two arms;
\item[3.] Connection to Bayesian persuasion and optimal persuasion for a special case;
\item[4.] How much information to reveal: recommendations vs. full revelation;
\item[5.] ``Hidden persuasion" : a general technique for Bayesian persuasion;
\item[6.] Incentivized exploration via ``hidden persuasion";
\item[7.] A necessary and sufficient assumption on the prior;
\item[8.] Discussion and literature review.
\end{OneLiners}


\end{abstract}

\section{Motivation and problem formulation}
\label{sec:intro}

Our motivation comes from recommendation systems. Users therein consume information from the previous users, and produce information for the future. For example, a decision to dine in a particular restaurant may be based on the existing reviews, and may lead to some new subjective observations about this restaurant. This new information can be consumed either directly (via a review, photo, tweet, etc.) or indirectly through aggregations, summarizations or recommendations, and can help others make similar choices in similar circumstances in a more informed way. This phenomenon applies very broadly, to the choice of a product or experience, be it a movie, hotel, book, home appliance, or virtually any other consumer's choice. Similar issues, albeit with higher stakes, arise in health and lifestyle decisions such as adjusting exercise routines or selecting a doctor or a hospital. Collecting, aggregating and presenting users' observations is a crucial value proposition of numerous businesses in the modern economy.


When self-interested individuals (\emph{agents}) engage in the information-revealing decisions discussed above, individual and collective incentives are misaligned. If a social planner were to direct the agents, she would trade off exploration and exploitation so as to maximize the social welfare. Absent such social planner, each agent's incentives are typically skewed in favor of exploitation. This is because agents tend to be myopic in their decisions, and (even when they are somewhat forward-looking) they prefer to side-step the costs of exploration and instead benefit from exploration done by others. Therefore, the society as a whole may suffer from insufficient amount of exploration. In particular, if a given alternative appears suboptimal given the information available so far, however sparse and incomplete, then this alternative may remain unexplored forever.

Let us consider a simple example in which the agents fail to explore. Suppose there are two actions $a\in\{1,2\}$ with deterministic rewards $\mu_1, \mu_2$ that are initially unknown. Each $\mu_a$ is drawn independently from a known Bayesian prior such that $\E[\mu_1]> \E[\mu_2]$. Agents arrive sequentially: each agent chooses an action, observes its reward and reveals it to all subsequent agents. Then the first agent chooses action $1$ and reveals $\mu_1$. If $\mu_1>\E[\mu_2]$, then all future agents also choose arm $1$. So, action $2$ never gets chosen, even though it may be better. This is particularly wasteful if the prior assigns a large probability to the event $\{\mu_2 \gg \mu_1>\E[\mu_2]\}$.

The problem of \textbf{incentivized exploration}
\index{incentivized exploration}
asks how to incentivize the agents to explore. We consider a \emph{principal} who cannot control the agents, but can communicate with them, \eg recommend an action and observe the outcome later on.
Such a principal would typically be implemented via a website, either one dedicated to recommendations and feedback collection (\eg Yelp, Waze), or one that actually provides the product or experience being recommended (\eg Netflix, Amazon). While the principal would often be a for-profit company, its goal for our purposes would typically be well-aligned with the social welfare.

We posit that the principal creates incentives \emph{only} via communication, rather monetary incentives such as rebates or discounts. What makes it feasible is \idf{information asymmetry}: the principal collects observations from the past agents and therefore has more information than any one agent. Aware of this information asymmetry, agents realize that they may benefit from algorithm's advice. In particular, they may be incentivized to follow the principal's recommendations, even if these recommendations sometimes include exploration.

Incentivizing exploration is a non-trivial task even in the simple example described above, and even if there are only two agents. This is \emph{Bayesian persuasion}, a well-studied problem in theoretical economics. When rewards are noisy, incentivized exploration is non-trivial even when incentives are not an issue. This is \emph{multi-armed bandits}, a well-studied problem in machine learning. Incentivized exploration addresses both problems simultaneously.





\xhdr{Problem formulation.}
Let us put forward a concrete problem formulation which we consider throughout this chapter. While idealized, it captures the essence of incentivized exploraton. There are $T$ rounds and two possible actions, denoted, resp.,
    $t\in [T]$ and $a\in \{1,2\}$.
The actions are also referred to as \emph{arms}, we use these terms interchangeably. In each round $t$, the principal chooses a message $\msg_t$ from some fixed universe $\msgSpace$. Then a new agent arrives, call it agent $t$, observes this message, chooses an arm $a_t$, and collects a reward $r_t\in [0,1]$ for this arm. The action and the reward are observed by the principal, but not by the other agents.


\begin{BoxedProblem}{Incentivized exploration}
\noindent In each round $t = 1,2,3 \LDOTS T $:
\vspace{-2mm}
\begin{enumerate}
  \item Principal chooses its message $\msg_t \in \msgSpace$.
  \item Agent $t$ arrives, observes $\msg_t$, and chooses an arm $a_t \in \{1,2\}$.
  \item (Agent's) reward $r_t\in [0,1]$ is realized.
  \item Action $a_t$ and reward $r_t$ are observed by the principal.
\end{enumerate}
\vspace{-2mm}
\end{BoxedProblem}

A given round reveals a reward $r_t$ for the chosen arm and no other information. In particular, the reward for the other arm -- had this arm been chosen by the agent -- is not revealed. This is precisely what necessitates exploration.

The rewards are generated as follows. A mean reward vector $\mu\in[0,1]^2$ is drawn from a Bayesian prior $\mP$ before the game starts, where $\mu_a$ is the mean reward of each arm $a\in\{1,2\}$. Each time an arm is chosen, the reward is realized as an independent draw from some fixed distribution specific to this arm. Specifically, there is a parameterized family
    $(\mD_x:\;x\in [0,1])$
of reward distributions such that $\E[\mD_x]=x$, and the reward distribution for each arm $a$ is defined as $\mD_x$ with $x=\mu_a$. The distribution family and the prior $\mP$ are known to the principal and the agents, whereas the mean reward vector $\mu$ is not known to anybody.

Let
    $\mu_a^0 = \E[\mu_a]$
denote the prior mean reward for arm $a$.
W.l.o.g., we assume that $\mu_1^0 \geq \mu_2^0$, \ie arm $1$ is weakly preferred according to the prior. We allow \emph{correlated priors}, \ie random variables $\mu_1, \mu_2$ can be correlated. An important special case is \emph{independent priors}, when $\mu_1, \mu_2$ are mutually independent. The paradigmatic reward distributions are \emph{deterministic rewards} (when $\mD_x$ always returns $x$) and \emph{Bernoulli rewards} (when $\mD_x$ is a Bernoulli distribution).

Agents choose arms so as to maximize their conditional expected rewards. Let us unpack this statement carefully. Each agent $t$ knows the messaging policy and the round she arrives in. As mentioned above, she knows the prior and the distribution family, but not the mean reward vector $\mu$. She does not directly observe anything from the previous rounds, except via the principal's message $\msg_t$. She strives to maximize $\E[r_t]$ given all available information. This task is ill-defined unless one specifies what the previous agents do; so, she assumes that all previous agents use the same decision rule. Formally, we define agents' behavior recursively:
\begin{align}\label{eq:agents-IE}
a_t = \argmaxARMS{ \E\sbr{ \mu_a \mid \mE_{t-1}, \msg_t  }},
\end{align}
where $\mE_{t-1}$ is the event that \refeq{eq:agents-IE} holds for all previous agents $s<t$. The $\min$ provides that the ties in $\argmax$ are always broken in favor of arm $1$ (and this is known to the principal and the agents). We posit this tie-breaking rule throughout to simplify presentation. Note that it only makes exploration more difficult.

During the game, the principal chooses messages $\msg_t\in \msgSpace$ according to some algorithm called the \emph{messaging policy}. \index{messaging policy} The latter is chosen by the principal before the game, along with the universe $\msgSpace$ (as a superset of all possible messages).

The principal's objective is to maximize the total reward
    $\REW := \sum_{t\in [T]} r_t$;
we make this objective more precise later, see \refeq{eq:regret}. One could also consider a more basic and immediate objective: sample arm $2$; we call it the \emph{pure exploration} objective, as it neglects exploitation. There are several ways to formalize it, \eg minimize the time horizon $T$ so as to guarantee that arm $2$ is sampled at least once.

\xhdr{Preliminaries: conditional expectations.}
For a more elementary exposition, realized rewards and messages can only take finitely many values. Then the notion of conditional expectation simplifies as follows. Let $X$ be a real-valued random variable, and let $Y$ be another random variable with arbitrary (not necessarily real-valued) but finite support $\mathcal{Y}$. Conditional expectation of $X$ given $Y$ is itself a random variable:
    $\E[X| Y] := f(Y)$,
where
    $f(y) = \E[X\mid Y=y]$
for all $y\in \mathcal{Y}$.
Conditional expectation given an event $E$ can be expressed as
    $\E[X | E] = \E[X | \indE{E}]$.

We often use the following fact, (a version of) the \emph{Law of Iterated Expectation}.

\begin{fact}\label{BIC:fact:iterated-expectations}
Suppose $X,Y$ are as above, and random variable $Z$ is determined by $Y$ and some other random variable $Z_0$ such that $X$ and $Z_0$ are independent (\eg the algorithm's random seed). Then
$ \E \sbr{ \E[X|Y] \mid Z } = \E[X|Z]$.
\end{fact}




\section{Connection to multi-armed bandits}
\label{sec:bandits}

In the ``social planner" version of incentivized exploration, an algorithm directly chooses the arm $a_t$ in each round $t$, and balances exploration and exploitation. This is a well-studied problem called \emph{multi-armed bandits},\index{multi-armed bandits} or \emph{bandits} for short.

\begin{BoxedProblem}{Multi-armed bandits}
\noindent In each round $t = 1,2,3 \LDOTS T $:
\begin{OneLiners}
  \item[] algorithm chooses an arm $a_t \in \{1,2\}$ and collects reward $r_t\in [0,1]$.
\end{OneLiners}
\end{BoxedProblem}

Bandit problems arise in recommendation systems as well as many other application domains. The problem name comes from a whimsical gambling scenario in which a gambler faces several identically-looking slot machines, a.k.a. one-armed bandits. The initial motivation for bandit problems came from the design of ``ethical" medical trials, which attain useful scientific data while minimizing harm to the patients. Prominent modern applications concern the Web: from tuning the look and feel of a website, to choosing which content to display or highlight, to optimizing web search results, to placing ads on webpages. A cluster of applications pertains to economics: a frequent seller (resp., buyer) can optimize prices and the selection of products (resp., offers); a frequent auctioneer (resp., bidder), particularly in ad auctions, can adjust its auction (resp., bids) over time; an online labor market can improve the assignment of tasks, workers and prices. A computer system can experiment and learn over time, rather than rely on a rigid design, so as to optimize the protocols inside a computer, a data center, or a communication network.

These applications motivate a rich problem space, with many dimensions along which the models can be made more expressive and closer to reality. To take a few examples: there can be many arms and a known, helpful structure that binds across arms; reward distributions may change over time; the algorithm may observe auxiliary payoff-relevant information before and/or after it chooses an action; the algorithm may be bound by global supply/budget constraints. The model we consider here -- with i.i.d. rewards and no ``extras" -- is the basic version.

A standard performance measure is \emph{regret},\index{regret} which compares the algorithm's reward to that of the best arm. Specifically, it measures how much one \emph{regrets} not knowing the best arm in advance, in expectation over the random rewards:
\begin{align}\label{eq:regret}
R(T) := R(T\mid \mu) :=  T\cdot\max(\mu_1,\mu_2) - \E[\REW].
\end{align}
We think of the mean reward vector $\mu$ as a problem instance. Typical results upper-bound regret uniformly (\ie in the worst case) over all problem instances, as a function of the time horizon $T$ and possibly also some parameters of the problem instance. For such results, the algorithm does not need to know the prior. Some results consider \emph{Bayesian regret},\index{Bayesian regret} \ie regret in expectation over the Bayesian prior, defined as
$\BReg(T) := \E_{\mu\sim\PP}\sbr{ R(T\mid \mu) } $.

While regret tracks cumulative performance over time, instantaneous performance can be measured by \emph{instantaneous regret} (a.k.a., \emph{simple regret}) at time $t$, defined as
    $\IReg(t) := \max(\mu_1,\mu_2) - \E[r_t]$.
That is, $\IReg(t)$ directly tracks how fast the algorithm converges to the best possible reward. Note that $R(T) = \sum_{t\in [T]} \IReg(t)$.

Going back to incentivized exploration, we observe that the messaging policy induces a bandit algorithm. (Indeed, the principal jointly with the agents constitute an algorithm that sequentially chooses actions and observes rewards according to the protocol of multi-armed bandits, and this algorithm's behavior is determined by the messaging policy.) Thus, one can directly compare messaging policies to bandit algorithms, using standard performance measures for the latter.


\subsection{Optimal exploration for two-armed bandits}
\label{sec:bandits-optimal}

We optimally solve the ``social planner problem" described above, \ie the two-armed bandit problem with i.i.d. rewards. Along the way, we invoke some of the central ideas from the literature on multi-armed bandits, and discuss what types of guarantees one is looking for (which is somewhat subtle). Our solution can be used as a subroutine for incentivized exploration, as we explain in Section~\ref{sec:repeatedHP}.


First, let us handle the randomness in rewards by arguing that average rewards are close to their respective means; such arguments are known as \emph{concentration inequalities}. Let $n_{t,a}$ be number of times a given arm $t$ has been chosen before round $t$, and let $\bar{\mu}_{t,a}$ be the average reward of this arm in these rounds. Then
\begin{align}\label{eq:chernoff}
\Pr\sbr{ \abs{\bar{\mu}_{t,a} - \mu_a} \leq \rad_{t,a}} \geq 1-2/T^4
\eqWHERE \rad_{t,a} :=
\sqrt{2\log(T)/n_{t,a}}.
\end{align}
This holds for any bandit algorithm, as a direct application of \emph{Azuma-Hoeffding inequality} (a standard, generic concentration inequality), restated in a notation that is convenient for our purposes. Thus, with high probability it holds that
\begin{align}\label{eq:clean}
  \mu_a
    \in \sbr{ \LCB_{t,a},\,\UCB_{t,a} }
    := \sbr{ \bar{\mu}_{t,a} - \rad_{t,a},\; \bar{\mu}_{t,a} + \rad_{t,a} }
\end{align}
for all arms $a$ and all rounds $t$.
The interval in \refeq{eq:clean} has two key properties: it contains $\mu_a$ with high probability and it can be computed from data in round $t$. An interval with these two properties is called a \emph{confidence interval} for $\mu_a$ (at time $t$), and its endpoints are called, resp., upper/lower \emph{confidence bounds}.

We use a simple algorithm called \term{AdaptiveRace}, which alternates the arms until their confidence intervals are disjoint (so that one arm appears better with high confidence). More precisely: we choose an arm uniformly at random in each round $t$ until
    $\LCB_{t,a} > \UCB_{t,a'}$,
and use arm $a$ forever after.



\begin{theorem}\label{thm:adaptiveRace}
Let $\Delta = |\mu_1-\mu_2|$. For each round $t$,
\term{AdaptiveRace} attains
\begin{align*}
R(t)
    = O\rbr{\min\rbr{ \sqrt{t\cdot  \log T},\,\tfrac{1}{\Delta}\cdot \log T }}
\quad\text{and}\quad
\IReg(t)
    = O\rbr{\sqrt{t^{-1}\cdot  \log T}}.
\end{align*}
\end{theorem}

\newcommand{\raw}{{\term{raw}}}

\begin{proof}
The Azuma-Hoeffding inequality mentioned above also implies that
\begin{align}\label{eq:clean-counts}
\abs{n_{t,a}-t/2} \leq O\rbr{\sqrt{t\,\log T}}
\quad\forall\;\text{arms $a$, rounds $t$}
\end{align}
with probability at least $1-2\,T^{-4}$. It suffices to perform the analysis conditional on the high-probability event that \eqref{eq:clean} and \eqref{eq:clean-counts} hold. We are interested in the ``raw" expressions
    $R^\raw(t) := t\cdot \max(\mu_1,\mu_2)-\sum_{s\leq t} r_s$ and
    $\IReg^\raw(t) := \max(\mu_1,\mu_2) - r_t$
whose expectations give, resp., $R(t)$ and $\IReg(t)$.

Let $\tau$ be the last round in which we did \emph{not} invoke the stopping rule. In all rounds $t\leq \tau$ the arms' confidence intervals overlap, so
\begin{align}\label{eq:clean-gap}
 \Delta
    \leq 2\rbr{ \rad_{t,a} + \rad_{t,a'} }
    = O\rbr{\sqrt{t^{-1}\;\cdot \log T }}
\quad \forall\; \text{rounds $t\leq \tau$}.
\end{align}
The best arm is chosen in all rounds $t>\tau$, so
$\IReg^\raw(t) \leq O\rbr{\sqrt{t^{-1}\;\cdot \log T }}$
and
    \[ R^\raw(t) \leq \Delta\cdot \min(t,\tau) \leq O\rbr{\sqrt{t\,\cdot \log T }}.\]
To obtain the gap-dependent regret bound, we observe that \refeq{eq:clean-gap} implies
    $\tau \leq O(\Delta^{-2}\;\log T)$,
which in turn implies
    $R^\raw(t) \leq \Delta\cdot \tau \leq O(\Delta^{-1}\;\log T)$.
\end{proof}

\begin{remark}
The significance of regret bounds in bandits typically focuses on the ``main terms" in the regret bounds, ignoring the constants and the logarithmic terms (unless the logarithmic term \emph{is} the main term). Thus, we obtain $R(t) \sim \sqrt{t}$ and $\IReg(t) \sim 1/\sqrt{t}$ in the worst case. Moreover, we obtain $R(T)\sim \log(T)$ when the \emph{gap} $\Delta = |\mu_1-\mu_2|$ is constant, with a multiplier that scales as $1/\Delta$. All $\log(T)$ terms can be replaced with $\log(t)$ via a slightly more involved analysis.
\end{remark}

Regret bounds in Theorem~\ref{thm:adaptiveRace} are optimal, as per the lower bounds presented below (without a proof). These lower bounds are rather subtle to state, and come in two flavors. First, a given upper bound cannot be improved \emph{in the worst case}, \ie for any algorithm there exists a problem instance that fools this algorithm. In fact, there is a pair of problem instances one of which fools \emph{every} algorithm. Second, the logarithmic regret bound is optimal \emph{for any given problem instance}, provided that the algorithm is at least somewhat good overall.%
\footnote{Such assumption is needed to rule out trivial solutions, \eg an algorithm that always plays arm $1$ is optimal for all instances with best arm $1$.}

\begin{theorem}\label{thm:bandit-LBs}
Focus on Bernoulli rewards. Fix any bandit algorithm \ALG.
\begin{itemize}
\item[(a)] Fix time horizon $T$ and $\eps > c/\sqrt{T}$, for a large enough absolute constant $c$. Consider problem instances
        $\mu = (\nicefrac{1}{2},\,\nicefrac{1}{2}+\eps)$
    and
        $\mu = (\nicefrac{1}{2},\,\nicefrac{1}{2}-\eps)$.
    Then on (at least) one of these instances $\ALG$ suffers regret
    $R(T)\geq \Omega(\nicefrac{1}{\eps})$ and $\IReg(T)\geq \Omega(\eps)$.

\item[(b)] Consider $T=\infty$. Suppose for each problem instance and each $\alpha>0$ there exists a constant $C$ such that $R(t) \leq C\,t^\alpha$ for all rounds $t\in\N$. Then for any given problem instance there exists time $t_0$ such that
        $R(t) \geq \frac{\mu^*(1-\mu^*)}{\Delta} \log t$
    for all rounds $t>t_0$, where $\mu^* = \max(\mu_1,\mu_2)$ and $\Delta = |\mu_1-\mu_2|$.

\end{itemize}
\end{theorem}

Note that \term{AdaptiveRace} adapts its exploration schedule to the observations. This property, called \emph{adaptive exploration}, is necessary to achieve the regret bounds listed in Theorem~\ref{thm:adaptiveRace}. Otherwise one can only achieve regret
    $R(T) = \tildeO(T^{2/3})$
and
    $\IReg(t) = \tildeO(t^{-1/3})$.
Two paradigmatic \emph{non}-adaptive algorithms are as follows. \term{ExploreFirst} samples each arm for a predetermined number $N$ of rounds, then chooses an arm with a larger average reward and plays it forever after. A more robust version called \term{EpsilonGreedy} spreads exploration uniformly: in each round $t$, with probability $\eps_t$ the algorithm explores by choosing an arm uniformly at random, and otherwise it exploits by choosing an arm with a larger average reward. One can achieve
    $R(T) = \tildeO(T^{2/3})$
by setting, resp., $N = T^{2/3}$ and $\eps_t = t^{-1/3}$. The proof uses the same technique as Theorem~\ref{thm:adaptiveRace}, we leave it as an exercise.

There are three techniques that attain optimal regret bounds, and extend far beyond the basic case of two-armed bandits.
\term{SuccessiveElimination}, of which \term{AdaptiveRace} is a special case, eliminates an arm from consideration once it appears worse than some other arm with high confidence. \term{UCB1} always chooses an arm with the best upper confidence bound. \term{ThompsonSampling} in each round forms a posterior distribution $\mP'$ on $\mu$, samples a mean reward vector $\mu'\sim \mP'$, and chooses the best arm according to $\mu'$. The first two algorithms are easier to analyze  (focusing on a high-prob event, like in Theorem~\ref{thm:adaptiveRace}), whereas \term{ThompsonSampling} attains the $R(t) = O(\frac{1}{\Delta}\,\log t)$ regret bound with an optimal constant factor.


\section{Connection to Bayesian persuasion}
\label{sec:BP}

A single round of incentivized exploration can be seen as a standalone game between the principal and the agent. More precisely, let us focus on the pure exploration objective (sample arm $2$), and consider some round $t>1$. The principal observes what happened in the previous rounds (the \emph{history}) and chooses a message $\msg\in \msgSpace$. The agent observes the message, chooses an arm $\armBP\in\{1,2\}$ and collects expected reward $\mu_{\armBP}$.  Incentives are misaligned: the agent prefers an arm with a larger reward, whereas the principal prefers arm $2$ no matter what. Thus, we can model the principal's reward as $u_{\armBP}$, where the reward vector is $u = (0,1)$. As the principal's messaging policy is fixed and known, and  the past agents' behavior is specified by \refeq{eq:agents-IE}, a joint distribution over $\mu$ and the histories is well-defined and known to both the agent and the principal. The agent's choice simplifies to
\begin{align}\label{eq:agents-BP}
\armBP = \argmaxARMS{ \E\sbr{ \mu_a \mid \msg} }.
\end{align}

To make this game more generic, let us posit that the principal observes an initial signal $\iSig$ from some universe $\sigSpace$ of possible signals, and receives reward $u_{\armBP}$ from some reward vector $u\in [0,1]^2$. The triple
    $(\iSig,\mu,u)$
is drawn from some joint prior $\mP$ which is known to both the principal and the agent. During the game, the principal chooses the message $\msg\in \msgSpace$ given the initial signal $\iSig$ according to some (possibly randomized) rule called the \emph{messaging policy}. Before the game, the principal chooses the messaging policy and the universe $\msgSpace$ (as a superset of all possible messages), so as to maximize its expected reward $\E[u_{\armBP}]$.
The messaging policy is known to the agent (which is needed to make \eqref{eq:agents-BP} well-defined).

This single-round game, called \emph{Bayesian persuasion}, \index{Bayesian persuasion} is well-studied in theoretical economics as a simple (yet quite rich) model of persuasion via strategic communication. The principal can neither choose the arms nor modify their payoffs directly. Instead, the principal leverages the initial signal and
strategically chooses which information to reveal to the agent.

\begin{BoxedProblem}{Bayesian persuasion}
\begin{enumerate}
\item  Initial signal $\iSig\in\sigSpace$ and reward vectors $\mu,u\in[0,1]^2$ are drawn:\\ triple $(\iSig,\mu,u)$ is drawn from Bayesian prior $\mP$.
  \item Principal observes $\iSig\in\sigSpace$,
    then computes message $\msg \in \msgSpace$.
  \item Agent observes $\msg$, then chooses arm
    $\armBP \in \{1,2\}$ as per \refeq{eq:agents-BP}.
  \item Agent receives expected reward $\mu_{\armBP}$, principal receives reward $u_{\armBP}$.
\end{enumerate}
\vspace{-2mm}
\end{BoxedProblem}

Motivating examples come from a variety of domains, in addition to recommendation systems. In a criminal investigation, the evidence can be seen as a ``message" from the prosecutor to the judge, and the prosecutor can strategically choose which evidence to request or seek out, so as to maximize the probability of conviction. Product information can be seen as a ``message" from the producer to potential consumers, and the producer may be able to strategically choose which tests to perform, which types of statistics to report, or what to offer in a free trial. Several examples concern grading policies: here, a grade is interpreted as a ``message" from the grader to the world. To wit, a student's grade is a ``message" from the school to potential employers; employee's performance feedback is a ``message" from the employer to the employee; a person's credit score is a ``message" from the credit agency to potential lenders. In all these cases, a grading policy can be chosen so as to achieve desirable social effects. Finally, the government chooses which information must be disclosed to the public, and can choose disclosure policies anticipating the strategic response thereto. For example, this concerns disclosing health risks (\eg those from environment pollution, infectious diseases, or vaccines), product data (esp. about foods and medicines), and some aspects of law enforcement strategies (\eg in policing, tax fraud detection, or wildlife protection).

The basic formulation defined above can be extended in several directions. For example, there can be multiple senders, the sender(s) may observe their own private signals, and messaging policies may be restricted to have a particular shape (\eg grading policies must assign same or higher grade for better achievements).

\subsection{Optimal persuasion for a special case}
\label{sec:persuasion-basic}

\newcommand{\tH}{{\term{H}}}
\newcommand{\tL}{{\term{L}}}

\newcommand{\muH}{v_\tH}
\newcommand{\muL}{v_\tL}

\newcommand{\mean}{\term{mean}} 
\newcommand{\vb}{{\bf b}} 

We work out a special case of Bayesian persuasion, which is \emph{also} a special case of incentivized exploration (with deterministic rewards and $T=2$ rounds). We obtain an optimal solution for this special case, whereas general solutions for incentivized exploration are only approximate. On the technical level, we showcase a fundamental technique from Bayesian persuasion, whereby one directly optimizes over the agent's posterior beliefs rather than over the principal's messaging policies.

We consider Bayesian persuasion with the pure exploration objective:
\begin{align}\label{eq:BP-sup-raw}
\sup\cbr{ \Pr\sbr{\armBP=2}:\;\; \text{messaging policies $\pi$} }.
\end{align}
We assume that the principal has full knowledge of arm $1$, \ie that $\iSig=\mu_1$. This is equivalent to incentivized exploration with deterministic rewards, $T=2$ rounds, and the pure exploration objective. Indeed, then arm $1$ is chosen in round $t=1$, a messaging policy is only used in round $t=2$, and its input is precisely $\mu_1$.



To make the problem more tractable, we also posit independent priors and $\mu_1$ having only two possible values: $\mu_1\in \{\muL,\muH\}$ with
    $0\leq \muL<\muH\leq 1$.
Since the problem is trivial if $\mu_2^0\not\in (\muL,\muH) $, we focus on the case when
    $\muL \leq  \mu_2^0 \leq \mu_1^0 \leq \muH$.
We make these assumptions without further notice in the rest of this subsection.


We solve this problem as follows:

\begin{theorem}\label{thm:BP}
The supremum in \refeq{eq:BP-sup-raw} equals
    $(\muH-\mu_1^0)/(\muH-\mu_2^0)$.
\end{theorem}


In particular, we achieve a larger $\Pr\sbr{\armBP=2}$ compared to full revelation. Indeed, under full revelation ($\msg = \mu_1$) the agent chooses arm $2$ if and only if $\mu_1<\mu_2^0$, and one can calculate that
    $\Pr[\mu_1<\mu_2^0] = (\muH-\mu_1^0)/(\muH-\muL)$.

Next, we prove Theorem~\ref{thm:BP}. We assume w.l.o.g. that there are only two possible messages. This follows from Claim~\ref{cl:revelation}, a general fact proved in the next section.

We are interested in the posterior distribution on $\mu_1$ formed by the agent after observing the message $\msg$; we refer to this distribution as a \emph{belief}. Thus, a belief is a distribution on $\{\muL,\muH\}$ determined by the messaging policy $\pi$ and the realized message. To simplify notation, we identify a belief with probability it assigns to $\muH$; we call this probability the \emph{scalar belief}. We write
    \[ B^\pi  := \Pr\sbr{ \mu_1 = \muH \mid \msg } \in [0,1]. \]
Note that $B^\pi$ is a random variable on $[0,1]$ whose realization is the scalar belief. We interpret $B^\pi$ as a distribution over beliefs induced by policy $\pi$. Since there are only two possible messages, $B^\pi$ has support of size at most $2$.

Any distribution $B^\pi$ is consistent with the prior, in the sense that
    $\E[B^\pi] = \Pr[\mu_1 = \muH]$.
(To prove this, apply Fact~\ref{BIC:fact:iterated-expectations} to the indicator function of $\{\mu_1 = \muH\}$.).
A distribution over scalar beliefs with this property is called \emph{Bayes-plausible}.

Let $\mB$ be the set of all Bayes-plausible distributions with support size at most $2$.
In fact, any distribution $B\in\mB$ can be realized by some messaging policy.

\begin{claim}\label{cl:BG-sufficient}
Any distribution $B\in\mB$ equals $B^\pi$ for some messaging policy $\pi$.
\end{claim}

\begin{proof}
We specify the policy $\pi$ explicitly. This necessitates additional notation (only for this proof). Fix distribution $B\in\mB$ with over scalar beliefs $\{b_1,b_2\}$. For each $j\in\{1,2\}$, let $\vb_j$ be the belief that corresponds to $b_j$, \ie a distribution over
    $V = \{\muL,\muH\}$
such that
    $\vb_j(\muH) = b_j$ and $\vb_j(\muL) = 1-b_j$.
Let $B(\vb_j)$ be the probability assigned to this belief by distribution $B$.

Now we are ready to specify the messaging policy $\pi$. It has two possible messages:  $\msgSpace = \{1,2\}$. For each $v\in V$ and $j\in\{1,2\}$,
\[ \Pr\sbr{ \msg=j \mid \mu_1 = v}
    = \vb_j(v)\; B(\vb_j)/\Pr[\mu_1=v].
    \qedhere
\]
\end{proof}

Henceforth we optimize directly over $\mB$, rather than over the messaging policies.

Let us spell out this maximization problem in more precise terms. Given any scalar belief $b\in[0,1]$, let
    $\mean(b) := b\,\muH + (1-b)\,\muL $
denote the corresponding expectation, \ie mean reward of arm $1$; we call it the  \emph{mean belief}. The agent chooses arm $2$ if and only if
    $\mean(b) <\mu_2^0$.
For a given messaging policy $\pi$, the agent chooses arm $2$ with probability
    $\Pr\sbr{ \mean(B^\pi)<\mu_2^0 }$.
Thus, the maximization problem is
\begin{align}\label{eq:BP-sup}
    \sup_{B\in \mB}\; \Pr\sbr{ \mean(B)<\mu_2^0 }.
\end{align}

\begin{lemma}\label{lm:BP-sup}
The supremum in \eqref{eq:BP-sup} equals $(\muH-\mu_1^0)/(\muH-\mu_2^0)$.
\end{lemma}

\begin{proof}
Fix some distribution $B\in\mB$ with support
    $\{b_\tL, b_\tH\}$,
where
    $0\leq b_\tL \leq b_\tH \leq 1$,
and respective probabilities $\{B_\tL, B_\tH\}$.
Writing out Bayes-plausibility, we obtain
\begin{align}
B_\tL\, b_\tL + B_\tH\, b_\tH &= \Pr[\mu_1 = \muH],
    \label{eq:lm:BP-sup:plausibility}\\
B_\tL\; \mean(b_\tL) + B_\tH\;\mean(b_\tH)
    &= \mean\rbr{ \Pr[\mu_1 = \muH] } = \mu_1^0.
    \label{eq:lm:BP-sup:plausibility-means}
\end{align}
\refeq{eq:lm:BP-sup:plausibility-means} is obtained by applying $\mean(\cdot)$ to both sides of \refeq{eq:lm:BP-sup:plausibility-means}. The two equations are equivalent since $\mean(\cdot)$ is monotone.
\refeq{eq:lm:BP-sup:plausibility-means} implies that
\[ \muL \leq \mean(b_\tL) \leq  \mu_1^0 \leq \mean(b_\tH)  \leq \muH.\]

Now, if
    $\mean(b_\tL)\geq \mu_2^0$
then
    $\Pr\sbr{ \mean(B)<\mu_2^0 } = 0$.
So, w.l.o.g. we will assume
    $\mean(b_\tL)< \mu_2^0$
from here on. Then
    $\Pr\sbr{ \mean(B)<\mu_2^0 } = B_\tL$.
Thus, the maximization problem in \refeq{eq:BP-sup} can be rewritten as follows:
\begin{align}\label{eq:lm:BP-sup:equivalent}
    \sup\cbr{ B_\tL:\;\; B\in \mB,\; \mean(b_\tL)< \mu_2^0 }.
\end{align}
Equivalently, we maximize $B_\tL$ subject to two constraints: Bayes-plausibility, as expressed by \refeq{eq:lm:BP-sup:plausibility-means}, and $\mean(b_\tL)< \mu_2^0$. This is maximized when
    $\mean(b_\tL)\to \mu_2^0$
and
    $\mean(b_\tH)= \muH$.
Plugging this into \eqref{eq:lm:BP-sup:plausibility-means}, we obtain
    $B_\tL \to (\muH-\mu_1^0)/(\muH-\mu_2^0)$.
\end{proof}

This completes the proof of Theorem~\ref{thm:BP}. The supremum in \eqref{eq:BP-sup-raw} and \eqref{eq:BP-sup} is not attained, essentially because the tie-breaking in \eqref{eq:agents-BP} favors arm $1$. If the tie-breaking favored arm $2$ instead, the supremum would be attained by setting $\mean(b_\tL)= \mu_2^0$ in the proof of
Lemma~\ref{lm:BP-sup}.


\section{How much information to reveal?}
\label{sec:revelation}

How much information should the principal reveal to the agents? Consider two extremes: revealing the entire history and recommending an arm without providing any supporting information. The former does not work, and the latter suffices.

\xhdr{Recommendations.} We transition from arbitrary messages to recommendations as follows. A \emph{recommendation policy} is a messaging policy whose messages are arms: formally, a message at each round $t$ is $\rec_t\in \{1,2\}$. Given an arbitrary messaging policy, the induced recommendation policy is one in which $\rec_t$ equals the right-hand side of \refeq{eq:agents-IE}. Let us argue that the agents choose recommended actions. Formally, we require Bayesian incentive-compatibility:

\begin{definition}\label{def:BIC}
Let
    $\recE_{t-1} = \{a_s=\rec_s:\, s<t \}$
denote the event that recommendations were followed before round $t$.
A recommendation policy is called \emph{Bayesian incentive-compatible} (\emph{BIC}) if for all rounds $t$ the following property holds:
\begin{align}\label{BIC:eqn:bic-constraint}
\rec_t =  \argmaxARMS{ \E\sbr{ \mu_a \mid \rec_t,\, \recE_{t-1} }}.
\end{align}
\end{definition}

\begin{claim}\label{cl:revelation}
Given any messaging policy, the induced recommendation policy is BIC.
\end{claim}

\begin{proof}
If a recommendation policy is induced by a messaging policy, event $\recE_t$ coincides with event $\mE$ from \refeq{eq:agents-IE}. For each round $t$ and all arms $a\neq a'$, we immediately obtain a version with $\msg_t$ in the conditioning, by definition of $\rec_t$:
\begin{align}
0 &\leq  \E\rbr{ \mu_a-\mu_{a'} \mid \msg_t,\;\rec_t=a,\, \mE_{t-1} }
    \label{eq:cl:revelation-pf}\\
0 &\leq  \E\rbr{ \mu_a-\mu_{a'} \mid \rec_t=a,\, \mE_{t-1} }
    & \EqComment{by \eqref{eq:cl:revelation-pf} and  Fact~\ref{BIC:fact:iterated-expectations}}
    \nonumber
\end{align}
Finally, if we have a strict equality in the latter equation, for some $a\neq a'$, then we also have it in the former equation, in which case $\rec_t = 1$ by \refeq{eq:agents-IE}.
\end{proof}

The above argument follows a well-known technique from theoretical economics called Myerson's \emph{direct revelation principle}. The conclusion is very strong, since it allows us to focus on BIC recommendation policies without loss of generality. However, it relies on assumptions of Bayesian rationality and power to commit that are implicit in our model, see Section~\ref{sec:discussion} for discussion thereof.

We are only interested in BIC recommendation policies from here on. By Eqns.~ \eqref{eq:agents-IE} and \eqref{BIC:eqn:bic-constraint}, all agents comply with recommendations issued by such a policy. Accordingly, it is simply a bandit algorithm with an auxiliary BIC constraint.

\xhdr{Full revelation does not work.} Does the principal need to bother designing and deploying a bandit algorithm? What if the principal just reveals  the full history and lets the agents choose for themselves? Then the agents, being myopic, would follow a ``greedy" bandit algorithm which always ``exploits" and never ``explores".

Formally, suppose that in each round $t$, the message $\msg_t$  includes the history
    $H_t := \{ (a_s,r_s):\, s\in [t-1] \}$.
Posterior mean rewards are determined by $H_t$:
\[
    \E[ \mu_a \mid \msg_t ] = \E[ \mu_a \mid H_t ]
    \quad\text{for each arm $a$}.
\]
(Because $\msg_t$ can only depend on $H_t$ and the algorithm's random seed.) Then
\begin{align}\label{BIC:eq:BG}
     a_t = \argmaxARMS{ \E\sbr{\mu_a \mid H_t} }.
\end{align}
This is (a version of) the greedy bandit algorithm, call it \greedy.

This algorithm performs terribly on a variety of problem instances, suffering Bayesian regret $\Omega(T)$. Whereas bandit algorithms can achieve regret $\tildeO(\sqrt{T})$ on all problem instances, as we saw in Section~\ref{sec:bandits-optimal}. The root cause of this inefficiency is that \greedy may never try arm $2$. For the special case of deterministic rewards, this happens with probability $\Pr[\mu_1\leq \mu_2^0]$, since $\mu_1$ is revealed in round $1$ and arm $2$ is never chosen if $\mu_1\leq \mu_2^0$; we discussed this case as a ``simple example" in Section~\ref{sec:intro}. This result (with a different probability) carries over to the general case.

\begin{theorem}\label{thm:BG}
With probability at least $\mu_1^0-\mu_2^0$, \greedy never chooses arm $2$.
\end{theorem}

\begin{proof}
In each round $t$, the key quantity is
    $Z_t = \E[ \mu_1-\mu_2 \mid H_t ]$.
Indeed, arm $2$ is chosen if and only if $Z_t<0$. Let $\tau$ be the first round when \greedy chooses arm $2$, or $T+1$ if this never happens. We use martingale techniques to prove that
\begin{align}\label{eq:thm:BG-Bayes-OST}
\E[Z_\tau] = \mu_1^0-\mu_2^0.
\end{align}

We obtain \refeq{eq:thm:BG-Bayes-OST} via a standard application of the optional stopping theorem (it can be skipped by readers who are not familiar with martingales). We observe that $\tau$ is a stopping time relative to
    $\mH = \rbr{H_t:\, t\in [T+1]}$,
and $\rbr{ Z_t: t\in [T+1]}$ is a martingale relative to $\mH$.
\footnote{The latter follows from a general fact that sequence
    $\E[X\mid H_t]$, $t\in [T+1]$
is a martingale w.r.t. $\mH$ for any random variable $X$ with $\E\sbr{|X|}\infty$. It is known as \emph{Doob martingale} for $X$.}
The optional stopping theorem asserts that $\E[Z_\tau] = \E[Z_1]$ for any martingale $Z_t$ and any bounded stopping time $\tau$. \refeq{eq:thm:BG-Bayes-OST} follows because
    $\E[Z_1] = \mu_1^0-\mu_2^0$.

On the other hand, by Bayes' theorem it holds that
\begin{align}\label{eq:thm:BG-Bayes}
\E[Z_\tau]
    = \Pr[ \tau\leq T ]\,\E[ Z_\tau \mid \tau\leq T ]
        + \Pr[ \tau>T ]\,\E[ Z_\tau \mid \tau>T ]
\end{align}
Recall that $\tau\leq T$ implies that $\greedy$ chooses arm $2$ in round $\tau$, which in turn implies that $Z_\tau \leq 0$ by definition of \greedy. It follows that
    $\E[ Z_\tau \mid \tau\leq T ]\leq 0$.
Plugging this into \refeq{eq:thm:BG-Bayes}, we find that
\[ \mu_1^0-\mu_2^0 = \E[Z_\tau] \leq \Pr[\tau>T].  \]
And $\{\tau>T\}$ is precisely the event that \greedy never tries arm 2.
\end{proof}

Under some mild assumptions, the algorithm never tries arm $2$ \emph{when it is in fact the best arm}, leading to $\Omega(T)$ Bayesian regret.

\begin{corollary}\label{cor:BG}
Consider independent priors such that $\Pr[\mu_1=1]<(\mu_1^0-\mu_2^0)/2$. Pick any $\alpha>0$ such that
    $\Pr[\mu_1\geq 1-2\,\alpha] \leq (\mu_1^0-\mu_2^0)/2$.
Then \greedy suffers Bayesian regret at least
    $T\cdot \rbr{ \nicefrac{\alpha}{2}\;(\mu_1^0-\mu_2^0) \; \Pr[\mu_2>1-\alpha] }$.
\end{corollary}

\begin{proof}
Let $\mE_1$ be the event that $\mu_1<1-2\alpha$ and \greedy never chooses arm $2$. By Theorem~\ref{thm:BG} and the definition of $\alpha$, we have
    $\Pr[\mE_1]\geq (\mu_1^0-\mu_2^0)/2$.

Let $\mE_2$ be the event that $\mu_2>1-\alpha$. Under event $\mE_1\cap \mE_2$, each round contributes
    $\mu_2-\mu_1\geq \alpha$
to regret, so
    $ \E[R(T) \mid \mE_1\cap \mE_2] \geq \alpha\,T$.

Since event $\mE_1$ is determined by the prior on arm $1$ and the rewards of arm $2$, it is independent from $\mE_2$. It follows that
\begin{align*}
\E[R(T)]
    &\geq \E[R(T) \mid \mE_1\cap \mE_2] \cdot \Pr[\mE_1\cap \mE_2] \\
    &\geq \alpha T\cdot (\mu_1^0-\mu_2^0)/2\cdot\Pr[\mE_2]. \qedhere
\end{align*}
\end{proof}

Here's a less quantitative but perhaps cleaner implication:

\begin{corollary}\label{cor:BG-basic}
Consider independent priors. Assume that each arm's prior has a positive density. That is, for each arm $a$, the prior on $\mu_a\in[0,1]$ has probability density function that is strictly positive on $[0,1]$. Then \greedy suffers Bayesian regret at least $c_\mP\cdot T$, where the constant $c_\mP> 0$ depends only on the prior $\mP$.
\end{corollary}

\section{``Hidden persuasion" for the general case}
\label{sec:HP}


We introduce a general technique for Bayesian persuasion, called \HP,  which we then use to incentivize exploration. The idea is to \emph{hide a little persuasion in a lot of exploitation}. We define a messaging policy which
inputs a signal $\iSig\in \sigSpace$, and
outputs a recommended arm $\rec\in \{1,2\}$. With a given probability $\eps$, we recommend what we actually want to happen, as described by the (possibly randomized) \emph{target function}
    $\ExploreArm: \sigSpace \to \{1,2\}$.
Otherwise we \emph{exploit}, \ie choose an arm that
maximizes  $\E\sbr{\mu_a \mid \iSig}$.


\vspace{-4mm}

\LinesNotNumbered \SetAlgoLined
\begin{algorithm}[h]
\caption{\HP with signal $\iSig$.}
\label{BIC:alg:basic}
\DontPrintSemicolon
{\bf Parameters:} probability $\eps>0$,
    function $\ExploreArm: \sigSpace \to \{1,2\}$.\;
{\bf Input:} signal realization $S \in \sigSpace$.\;
{\bf Output:} recommended arm $\rec$.\;\vspace{2mm}
With probability $\eps>0$,
    \tcp*{persuasion branch}
\algTAB $\rec \leftarrow \ExploreFn[S]$ \;
else \tcp*{exploitation branch}
\algTAB
    $\rec \leftarrow
        \min\rbr{ \arg\max_{a\in \{1,2\}} \E\sbr{\mu_a \mid \iSig = S}}$
\end{algorithm}
\vspace{-2mm}


We are interested in the (single-round) BIC property:
\begin{align}\label{BIC:eqn:bic-basic}
\rec = \argmaxARMS{ \E\sbr{\mu_a \mid \rec} }.
\end{align}
We prove that \HP satisfies this property when the persuasion probability $\eps$ is sufficiently small, so that the persuasion branch is offset by exploitation. A key quantity here is a random variable which summarizes the meaning of $\iSig$:
\[ G := \E[\mu_2-\mu_1 \mid \iSig]
\qquad\qquad\EqComment{posterior gap}.
\]

\begin{lemma}\label{BIC:lm:basic}
\HP with persuasion probability $\eps>0$ is BIC, for any target function $\ExploreArm$,
as long as
    $\eps < \tfrac13\,\E\left[ G\cdot \ind{G>0} \right]$.
\end{lemma}

\begin{remark}\label{rem:HP-works}
A suitable $\eps>0$ exists if and only if $\Pr[G>0]>0$.
Indeed, if $\Pr[G>0]>0$ then $\Pr[G>\delta]=\delta'>0$ for some $\delta>0$, so
\begin{align*}
\E\left[ G\cdot \ind{G>0} \right]
    \geq \E\left[ G\cdot \ind{G>\delta} \right]
    = \Pr[G>\delta]\cdot \E[G\mid G>\delta]
    \geq  \delta\cdot\delta'>0.
\end{align*}
\end{remark}

The rest of this section proves Lemma~\ref{BIC:lm:basic}. To keep the exposition elementary, we assume that the universe $\sigSpace$ is finite.%
\footnote{Otherwise our proof requires a more advanced notion of conditional expectation.}
We start with an easy observation: for any algorithm, it suffices to guarantee the BIC property when arm $2$ is recommended.

\begin{claim}
Assume
    \eqref{BIC:eqn:bic-basic}
holds for arm $\rec=2$. Then it also holds for $\rec=1$.
\end{claim}

\begin{proof}
If arm $2$ is never recommended, then the claim holds trivially since
    $\mu_1^0 \geq \mu_2^0$.
Now, suppose both arms are recommended with some positive probability. Then
\begin{align*}
0   \geq \E[\mu_2 - \mu_1]
    = \textstyle \sum_{a\in\{1,2\}}\; \E[\mu_2-\mu_1 \mid \rec=a]\,\Pr[\rec=a].
\end{align*}
Since
    $\E[\mu_2-\mu_1 \mid \rec=2] > 0$
by the BIC assumption,
    $\E[\mu_2-\mu_1 \mid \rec=1] < 0$.
\end{proof}

Thus, we need to prove \eqref{BIC:eqn:bic-basic} for $\rec=2$, \ie that
\begin{align}\label{BIC:eqn:bic-basic-2}
    \E\sbr{ \mu_2-\mu_1 \mid \rec=2 } > 0.
\end{align}
(We note that $\Pr[\mu_2-\mu_1]>0$, \eg because $\Pr[G>0]>0$, as per Remark~\ref{rem:HP-works}).

Denote the event $\{\rec=2\}$ with $\mE_2$.
By Fact~\ref{BIC:fact:iterated-expectations},
    $ \E\sbr{ \mu_2-\mu_1 \mid \mE_2 }  = \E[G \mid \mE_2]$.
\footnote{This is the only step in the analysis where it is essential that both the persuasion and exploitation branches (and therefore event $\mE_2$) are determined by the signal $\iSig$.}

We focus on the posterior gap $G$ from here on. More specifically, we work with expressions of the form
    $ F(\mE) := \E\sbr{G\cdot \indE{\mE}}$,
where $\mE$ is some event. Proving \refeq{BIC:eqn:bic-basic-2} is equivalent to proving that
    $ F(\mE_2) > 0$;
we prove the latter in what follows.

We will use the following fact:
\begin{align}\label{BIC:eq:F-disjoint}
    F(\mE \cup \mE') = F(\mE) + F(\mE')
\quad \text{for any disjoint events $\mE,\mE'$}.
\end{align}
Letting $\ExploreE$ (resp., $\ExploitE$) be the event that the algorithm chooses persuasion branch (resp., exploitation branch), we can write
\begin{align}\label{BIC:eq:F-U-disjoint}
     F(\mE_2) = F(\ExploreE \eqAND \mE_2) + F(\ExploitE \eqAND \mE_2).
\end{align}

We prove that this expression is non-negative by analyzing the persuasion and exploitation branches separately. For the exploitation branch, the events
    $\{\ExploitE \eqAND \mE_2\}$
and
    $\{\ExploitE \eqAND G>0\}$
are the same by algorithm's specification. Therefore,
\begin{align*}
F(\ExploitE \eqAND \mE_2)
    &=  F(\ExploitE \eqAND G>0)\\
    &= \E[G \mid \ExploitE \eqAND G>0] \cdot \Pr[\ExploitE \eqAND G>0]
        &\EqComment{by definition of $F$} \\
    &= \E[G\mid G>0] \cdot \Pr[G>0]\cdot (1-\eps)
        &\EqComment{by independence} \\
    &= (1-\eps)\cdot F(G>0)
        &\EqComment{by definition of $F$}.
\end{align*}

For the persuasion branch, recall that $F(\mE)$ is non-negative for any event $\mE$ with $G\geq 0$, and non-positive for any event $\mE$ with $G\leq 0$. Therefore,
\begin{align*}
F(\ExploreE \eqAND \mE_2)
    &= F(\ExploreE \eqAND \mE_2 \eqAND G<0)
        + F(\ExploreE \eqAND \mE_2 \eqAND G\geq 0)
       &\EqComment{by \eqref{BIC:eq:F-disjoint}} \\
    &\geq F(\ExploreE \eqAND \mE_2 \eqAND G<0) \\
    &= F(\ExploreE \eqAND G<0) - F(\ExploreE \eqAND \neg \mE_2 \eqAND G<0)
        &\EqComment{by \eqref{BIC:eq:F-disjoint}}\\
    &\geq F(\ExploreE \eqAND G<0)\\
    &= \E[G\mid \ExploreE \eqAND G<0]\cdot \Pr[\ExploreE \eqAND G<0]
        &\EqComment{by defn. of $F$}\\
    &= \E[G\mid G<0]\cdot \Pr[G<0]\cdot \eps
        &\EqComment{by independence}\\
    &= \eps\cdot F(G<0)
    &\EqComment{by defn. of $F$}.
\end{align*}

Putting this together and plugging into \eqref{BIC:eq:F-U-disjoint}, we have
\begin{align}\label{BIC:eq:F-U}
F(\mE_2) \geq \eps\cdot F(G<0) + (1-\eps)\cdot F(G>0).
\end{align}

Now, applying \eqref{BIC:eq:F-disjoint} yet again we see that
   $ F(G<0) + F(G>0) = \E[\mu_2-\mu_1]$.
Plugging this back into \eqref{BIC:eq:F-U} and rearranging, it follows that $F(\mE_2)> 0$ whenever
\[ F(G>0) > \eps \rbr{ 2F(G>0)+\E[\mu_1-\mu_2] }. \]
In particular, $\eps< \tfrac13\cdot F(G>0)$ suffices. This completes the proof of Lemma~\ref{BIC:lm:basic}.

\section{Incentivized exploration via ``hidden persuasion"}
\label{sec:repeatedHP}

%
%

Let us develop the hidden persuasion technique into an algorithm for incentivized exploration. We take an arbitrary bandit algorithm \ALG, and consider a repeated version of \HP (called \RepeatedHP), where the persuasion branch executes one call to \ALG. We interpret calls to \ALG as exploration. To get started, we include $N_0$ rounds of ``initial exploration", where arm $1$ is chosen. The exploitation branch conditions on the history of all previous exploration rounds:
\begin{align}\label{eq:repeatedHP-history}
\mS_t = \rbr{ (s,a_s,r_s): \text{all exploration rounds $s<t$}}.
\end{align}

\LinesNotNumbered \SetAlgoLined
\begin{algorithm}[!h]
\caption{\RepeatedHP with bandit algorithm \ALG.}
\label{BIC:alg:reduction}
\DontPrintSemicolon
{\bf Parameters:} $N_0\in\N$, exploration probability $\eps>0$\;
In the first $N_0$ rounds, recommend arm $1$. \tcp*{initial exploration}
In each subsequent round $t$,\;
\algTAB With probability $\eps$ \tcp*{explore}
\algTAB \algTAB call \ALG, let $\rec_t$ be the chosen arm,
feed reward $r_t$ back to \ALG. \;
\algTAB else \tcp*{exploit}
\algTAB \algTAB
    $\rec_t \leftarrow
        \min\rbr{\arg\max_{a\in \{1,2\}} \E[\mu_a \mid \mS_t] }$.
        \tcp*{$\mS_t$ from \eqref{eq:repeatedHP-history}}
\end{algorithm}
\vspace{-4mm}


\begin{remark}
\RepeatedHP can be seen as a reduction from bandit algorithms to BIC bandit algorithms. The simplest version always chooses arm $2$ in exploration rounds, and (only) provides non-adaptive exploration. For better regret bounds, \ALG needs to perform adaptive exploration, as per Section~\ref{sec:bandits}.
\end{remark}

Each round $t>N_0$ can be interpreted as \HP with signal $\mS_t$,
where the ``target function" executes one round of algorithm \ALG. Note that $\rec_t$ is determined by $\mS_t$ and the random seed of \ALG, as required by the specification of \HP. Thus, Lemma~\ref{BIC:lm:basic} applies, and yields the following corollary in terms of
    $G_t = \E[\mu_2-\mu_1 \mid \mS_t]$,
the posterior gap given signal $\mS_t$.

\begin{corollary}\label{BIC:cor:basic}
\RepeatedHP is BIC if
    $\eps < \tfrac13\,\E\left[ G_t\cdot \ind{G_t>0} \right]$
for each time $t>N_0$.
\end{corollary}

For the final BIC guarantee, we show that it suffices to focus on $t=N_0+1$.

\begin{theorem}\label{BIC:thm:reduction-BIC}
\RepeatedHP with exploration probability $\eps>0$ and $N_0$ initial samples of arm $1$ is BIC as long as
    $\eps < \tfrac13\,\E\left[ G\cdot \ind{G>0} \right]$,
where $G = G_{N_0+1}$.
\end{theorem}

\begin{proof}
The only remaining piece is the claim that the quantity
    $\E\left[ G_t\cdot \ind{G_t>0} \right]$
does not decrease over time. This claim holds for any sequence of signals
    $(\mS_1,\mS_2 \LDOTS S_T)$
such that each signal $S_t$ is determined by the next signal $S_{t+1}$.

Fix round $t$. Applying Fact~\ref{BIC:fact:iterated-expectations} twice, we obtain
\begin{align*}
\E[G_t \mid G_t >0]
    = \E[\mu_2-\mu_1 \mid G_t>0]
    = \E[G_{t+1} \mid G_t >0].
\end{align*}
(The last equality uses the fact that $S_{t+1}$ determines $S_t$.) Then,
\begin{align*}
\E\left[ G_t\cdot \ind{G_t>0} \right]
    &= \E[G_t \mid G_t >0]\cdot \Pr[G_t>0] \\
    &= \E[G_{t+1} \mid G_t >0]\cdot \Pr[G_t>0] \\
    &= \E\left[ G_{t+1}\cdot \ind{G_t>0} \right] \\
    &\leq \E\left[ G_{t+1}\cdot \ind{G_{t+1}>0} \right].
\end{align*}
The last inequality holds because
    $ x\cdot \ind{\cdot} \leq x\cdot \ind{x>0}$
for any $x\in R$.
\end{proof}


\begin{remark}\label{rem:repeatedHP-condition}
The theorem focuses on the posterior gap $G$ given $N_0$ initial samples from arm $1$.
The theorem requires parameters $\eps>0$ and $N_0$ to satisfy some condition that depends only on the prior. Such parameters exist, if and only if $\Pr[G>0]>0$ for some $N_0$. (This is for precisely the same reason as in Remark~\ref{rem:HP-works}.) The latter condition is in fact necessary, as we will see in Section~\ref{sec:fighting-chance}.
\end{remark}

\newcommand{\uN}{\underline{N}}

Performance guarantees for \RepeatedHP completely separated from the BIC guarantee, in terms of results as well as proofs. Essentially, \RepeatedHP learns at least as fast as an appropriately slowed-down version of \ALG. There are several natural ways to formalize this, in line with the standard performance measures for multi-armed bandits. For notation,
let $\REW^{\ALG}(n)$ be the total reward of \ALG in the first $n$ rounds of its execution, and let $\BReg^{\ALG}(n)$ be the corresponding Bayesian regret.

\begin{theorem}\label{thm:reduction-perf}
Consider \RepeatedHP with exploration probability $\eps>0$ and $N_0$ initial samples.
 Let $N$ be the number of exploration rounds $t>N_0$. %
\footnote{Note that $\E[N] = \eps(T-N_0)$, and
    $|N-\E[N]|\leq O(\sqrt{T\,\log T})$
with high probability.}
Then:

\begin{itemize}
\item[(a)] If $\ALG$ always chooses arm $2$,  \RepeatedHP chooses arm $2$ at least N times.

\item[(b)] The expected reward of  \RepeatedHP is at least
        $\tfrac{1}{\eps}\, \E\sbr{ \REW^{\ALG}(N) }$.

\item[(c)] Bayesian regret of  \RepeatedHP is
        $ \BReg(T) \leq N_0+\tfrac{1}{\eps}\, \E\sbr{ \BReg^{\ALG}(N) }$.

\end{itemize}
\end{theorem}

\begin{proof}[Proof Sketch]
Part (a) is obvious. Part (c) trivially follows from part (b). The proof of part (b) invokes Wald's identify and the fact that the expected reward in ``exploitation" is at least as large as in ``exploration" for the same round.
\end{proof}



We match the Bayesian regret of \ALG up to by factors $N_0,\,\tfrac{1}{\eps}$ that depend only on the prior. So, we match the optimal regret for a given prior $\mP$ if \ALG is optimal for $\mP,T$. We achieve the $\tildeO(\sqrt{T})$ regret bound for all problem instances, \eg using \term{AdaptiveRace} from Section~\ref{sec:bandits}. The prior-dependent factors can be arbitrarily large, depending on the prior; we interpret this as the ``price of incentives".

\section{A necessary and sufficient assumption on the prior}
\label{sec:fighting-chance}

We need to restrict the prior $\mP$ so as to give the algorithm a fighting chance to convince some agents to try arm $2$. (Recall that $\mu_1^0\geq \mu_2^0$.) Otherwise the problem is just hopeless.
For example, if $\mu_1$ and $\mu_1-\mu_2$ are independent, then samples from arm $1$ have no bearing on the conditional expectation of $\mu_1-\mu_2$, and therefore cannot possibly incentivize any agent to try arm $2$.

We posit that arm $2$ \emph{can} appear better after seeing sufficiently many samples of arm $1$. Formally, we consider the posterior gap given $n$ samples from arm $1$:
\begin{align}\label{BIC:eq:gap1}
    G_{1,n} := \E[\, \mu_2-\mu_1 \mid \samples{n} ],
\end{align}
where $\samples{n}$ denotes an ordered tuple of $n$ independent samples from arm $1$.
We focus on the property that this random variable can be positive:
\begin{align}\label{BIC:eq:prop}
\Pr\left[ G_{1,n} > 0 \right]>0 \quad
\text{for some prior-dependent constant $n=n_\mP<\infty$}.
\end{align}
For independent priors, this property can be simplified to
    $\Pr[\mu_2^0>\mu_1]>0$.
Essentially, this is because
    $G_{1,n} = \mu_2^0 - \E[\mu_1 \mid \samples{n}]\to \mu_2^0 - \mu_1$.

We prove that Property~\eqref{BIC:eq:prop} is necessary for BIC bandit algorithms. Recall that it is sufficient for \RepeatedHP, as per Remark~\ref{rem:repeatedHP-condition}.


\begin{theorem}\label{BIC:thm:LB}
Absent \eqref{BIC:eq:prop}, any BIC algorithm never plays arm $2$.
\end{theorem}

\begin{proof}
Suppose \propref{BIC:eq:prop} does not hold. Let $\ALG$ be a strongly BIC algorithm.
We prove by induction on $t$ that $\ALG$ cannot recommend arm $2$ to agent $t$.

This is trivially true for $t=1$. Suppose the induction hypothesis is true for some $t$. Then the decision whether to recommend arm $2$ in round $t+1$ (\ie whether $a_{t+1}=2$) is determined by the first $t$ outcomes of arm $1$ and the algorithm's random seed. Letting $U=\{a_{t+1}=2\}$, we have
\begin{align*}
\E[\mu_2-\mu_1 \mid U]
    &= \E\left[\;\; \E[\mu_2-\mu_1 \mid \samples{t}]\;\; \mid U \right]
        &\EqComment{by Fact~\ref{BIC:fact:iterated-expectations}}\\
    &= \E[G_{1,t}\mid U]
        &\EqComment{by definition of $G_{1,t}$}\\
    &\leq 0
        &\EqComment{since \eqref{BIC:eq:prop} does not hold}.
\end{align*}
The last inequality holds because the negation of \eqref{BIC:eq:prop} implies
    $\Pr[G_{1,t}\leq 0]=1$.
This contradicts \ALG being BIC, and completes the induction proof.
\end{proof}

%
%

\section{Discussion and literature review}
\label{sec:discussion}

To conclude this chapter, we include relevant citations, survey the ``landscape" of incentivized exploration, and place this work in a larger context. A more detailed literature review can be found in \citep[Ch. 11.6]{slivkins-MABbook}.

The technical results in this chapter map to the literature as follows. The results on incentivized exploration in  Sections~\ref{sec:HP}-\ref{sec:fighting-chance} are from
\citet{ICexploration-ec15}. The algorithmic results on bandits presented in  Section~\ref{sec:bandits} are folklore; the lower bounds trace back to
\citet{Lai-Robbins-85,bandits-exp3}. The Bayesian persuasion example in Section~\ref{sec:BP}
follows the technique from \cite{Kamenica-aer11}.
Inefficiency of \greedy has been folklore for decades. The general impossibility result (Theorem~\ref{thm:BG}) first appeared in \citep[Ch. 11]{slivkins-MABbook}, and
is due to \citet{GreedyFails19}.

Incentivized exploration was introduced in \citep{Kremer-JPE14,Che-13}. Our model of incentivized exploration was studied, and largely resolved, in \citep{Kremer-JPE14,ICexploration-ec15,ICexplorationGames-ec16,Selke-PoIE-ec21}.
Results come in a variety of flavors: to wit,
optimal policies for deterministic rewards;
``frequentist" regret bounds (as in Theorem~\ref{thm:adaptiveRace});
extension to $K>2$ arms;
performance loss compared to bandits;
and exploring all actions that can possibly be explored. \emph{Algorithms} come in a variety of flavors, too. In particular, \RepeatedHP is made more efficient by inserting a third ``branch" that combines exploration and exploitation, and allows the exploration probability to increase substantially over time. Moreover, \term{ThompsonSampling}, an optimal bandit algorithm mentioned in Section~\ref{sec:bandits}, is BIC for independent priors when initialized with enough samples of each arm. Unlike \RepeatedHP, this algorithm does not suffer from multiplicative prior-dependent blow-up; however, the initial samples should be collected by some other method.


Our model can be made more realistic in three broad directions. First, one can generalize the \emph{exploration} problem, in all ways that one can generalize multi-armed bandits. In particular, \RepeatedHP generalizes well, because it inputs a generic exploration algorithm \ALG \citep{ICexploration-ec15}. However, this approach does not scale to ``structured" exploration problems with exponentially many actions or policies. \cite{IncentivizedRL} address this issue for incentivized reinforcement learning, via a different technique.

Second, one can generalize the \emph{persuasion} problem in our model, in all ways that one can generalize Bayesian persuasion. Prior work considers repeated games and misaligned incentives \citep{ICexplorationGames-ec16},
heterogenous agents \citep{Jieming-multitypes18},
partially known agents' beliefs \citep{ICexploration-ec15,Jieming-multitypes18},
and unavoidable information leakage \citep{Bahar-ec16,Bahar-ec19}.

Third, one can relax the standard (yet strong) economic assumptions that the principal commits to the messaging policy, and the agents
optimize their Bayesian-expected rewards. The latter assumption rules out generic issues such as risk aversion, probability matching, or approximate reasoning, as well as problem-specific issues such as aversion to being singled out for exploration, or reluctance to follow a recommendation without supporting evidence.
\cite{Jieming-unbiased18} achieve near-optimal regret under weaker assumptions: they only need to specify how agents respond to fixed datasets and allow a flexible ``frequentist" response thereto.



Related, but technically different models of incentivized exploration feature
time-discounted utilities \citep{Bimpikis-exploration-ms17};
monetary incentives \citep{Frazier-ec14,Kempe-colt18};
continuous information flow \citep{Che-13}; and coordination of costly ``exploration decisions" which are separate from ``payoff-generating decisions" \citep{Bobby-Glen-ec16,Annie-ec18-traps,Liang-ec18}. \greedy (full revelation) works well for heterogenous agents, under strong assumptions on the structure of rewards and diversity of agent types \citep{kannan2018smoothed,bastani2017exploiting,Greedy-Manish-18,AcemogluMMO19}.

Incentivized exploration is closely related to two prominent subareas of theoretical economics. \emph{Information design}
\citep{BergemannMorris-survey19,Kamenica-survey19}
studies the design of information disclosure policies and incentives that they create. One fundamental model is Bayesian persuasion \citep{Kamenica-aer11}.
\emph{Social learning} \citep{Golub-survey16,Horner-survey16} studies self-interested agents that interact and learn over time in a shared environment.

Multi-armed bandits have been studied since 1950-ies in economics, operations research and computer science, with a big surge in the last two decades coming mostly from machine learning.
This vast literature is covered in various books, the most recent ones are
\citet{slivkins-MABbook} and \citet{LS19bandit-book}.

\newpage

\bibliography{bib-abbrv,bib-slivkins,bib-AGT,bib-bandits}\label{refs}
\bibliographystyle{cambridgeauthordate}

\end{document}